\documentclass[conference]{IEEEtran}

\usepackage{amsmath, amsthm, amssymb}
\usepackage{latexsym}
\usepackage{graphicx}
\usepackage{verbatim}
\usepackage{mathrsfs}
\usepackage{algorithmic}
\usepackage{float}
\usepackage[lofdepth, lotdepth]{subfig}
\usepackage{array}
\usepackage{url}

\theoremstyle{plain}
\newtheorem{theorem}{Theorem}
\newtheorem{lemma}[theorem]{Lemma}

\theoremstyle{definition}

\newtheorem{example}[theorem]{Example}

\newcommand{\A}{{\mathcal A}}

\newcommand{\I}{{\mathcal I}}

\newcommand{\M}{{\mathcal M}}

\newcommand{\bM}{{\boldsymbol{M}}}

\newcommand{\bA}{{\boldsymbol A}}

\newcommand{\bV}{\boldsymbol{V}}

\newcommand{\bu}{{\boldsymbol u}}

\newcommand{\bG}{{\boldsymbol{G}}}

\newcommand{\fq}{\mathbb{F}_q}

\newcommand{\al}{\alpha}

\newcommand{\supp}{{\sf supp}}
\newcommand{\sgn}{{\sf sgn}}

\newcommand{\define}{\stackrel{\mbox{\tiny $\triangle$}}{=}}

\newcommand{\nin}{\noindent}

\newcommand{\et}{{\emph{et al.}}}

\newcommand{\vM}{{{\sf var}}(\boldsymbol{M})}

\begin{document}
\pagestyle{empty}

\title{On the Existence of MDS Codes Over Small Fields With Constrained Generator Matrices}
\author{
   \IEEEauthorblockN{
     Son Hoang Dau\IEEEauthorrefmark{1},
     Wentu Song\IEEEauthorrefmark{2}, 
     Chau Yuen\IEEEauthorrefmark{3}
		} 
   \IEEEauthorblockA{
   Singapore University of Technology and Design, Singapore\\     
		Emails: $\{${\it\IEEEauthorrefmark{1}sonhoang\_dau, 
		\IEEEauthorrefmark{2}wentu\_song,
		\IEEEauthorrefmark{3}yuenchau}$\}$@sutd.edu.sg		
		}
 }
\maketitle

\begin{abstract}
We study the existence over small fields of Maximum Distance Separable (MDS) codes with generator matrices
having specified supports (i.e. having specified locations of zero entries). This problem unifies and simplifies the problems posed in recent works 
of Yan and Sprintson (NetCod'13) on weakly secure cooperative data exchange, of Halbawi {\et} (arxiv'13) on distributed
Reed-Solomon codes for simple multiple access networks, and of Dau {\et} (ISIT'13) on MDS codes with 
balanced and sparse generator matrices. 
We conjecture that there exist such $[n,k]_q$ MDS codes as long as $q \geq n + k - 1$, 
if the specified supports of the generator matrices
satisfy the so-called MDS condition, which can be verified in polynomial time.     
We propose a combinatorial approach to tackle the conjecture, and prove that the conjecture
holds for a special case when the sets of zero coordinates of rows of the generator matrix
share with each other (pairwise) at most one common element. 
Based on our numerical result, the conjecture is also verified for all $k \leq 7$. 
Our approach is based on a novel generalization of the well-known Hall's marriage theorem,
which allows (overlapping) multiple representatives instead of a single representative for each subset.     
\end{abstract}

\section{Introduction}
\label{sec:intro}

\subsection{A Conjecture on MDS Codes}

Maximum Distance Separable (MDS) codes~\cite{MW_S}, in particular Reed-Solomon (RS) codes, arguably form 
the most structurally elegant family of error-correcting codes in the literature of coding theory.
These well-known codes are ubiquitous, with applications found across a vast 
area of modern information technology, ranging from data storage media
such as CDs and DVDs and data storage systems such as RAID 6, to deep space communications.   

Despite a huge body of research on MDS codes, there are still 
challenging open problems, such as the one stated in the famous \emph{MDS Conjecture}: 
there exists an $[n,k]_q$ MDS code if and only if $n \leq q+1$ for all $q$ and $2 \leq k \leq q-1$, 
except when $q$ is even and $k \in \{3,q-1\}$, in which case $n \leq q + 2$. 
The existence of such an MDS code if the above conditions are satisfied is well known, 
via the use of (extended) Generalized Reed-Solomon (GRS) codes~\cite{MW_S}.
However, when we impose some further condition on the structure of the generator matrix, 
then the existence of such an MDS code over small fields, when $n$ and $k$ are fixed, is not known. 
In this paper we pose another conjecture (Fig.~\ref{fig:MDS_conjecture}) on the existence of MDS 
codes over small fields with some constraint on the support of the generator matrix. 
In this conjecture, if we allow the field size $q$ to be sufficiently large, then 
it is known that there exists an $[n,k]_q$ MDS code satisfying the stated requirement. 
A proof for
this claim can be found, for instance, in \cite{DauSongDongYuenISIT2013}, Lemma 1--4, 
with the condition that rows of $\bM$ have weight $n-k+1$ being removed in Lemma 1. 
However, requiring the field size to be as low as $n + k - 1$ makes the problem much more challenging. 
\begin{figure}[htb]
\centering{
\fbox{
\parbox{3.3in}{
\centerline{\bf GM-MDS Conjecture}
Let $\bM = (m_{i,j})$ be a $k \times n$ binary matrix satisfying the so-called \emph{MDS Condition}:
\[
|\cup_{i \in I} \supp(\bM_i)| \geq n - k + |I|,
\]
for all nonempty subsets $I \subseteq \{1,2,\ldots, k\}$, where $\supp(\bM_i) = \{j \mid 1 \leq j \leq n,\ m_{i,j} \neq 0\}$
is the support of the $i$th row of $\bM$. Then for every prime power $q \geq n + k - 1$, 
there exists an $[n,k]_q$ MDS code that has a generator matrix $\bG = (g_{i,j})$ satisfying 
$g_{i,j} = 0$ whenever $m_{i,j} = 0$ (we say that $\bG$ \emph{fits} $\bM$ in this case). 
}
}
}
\caption{The GM-MDS Conjecture}
\label{fig:MDS_conjecture}
\vspace{-10pt}
\end{figure}

\subsection{Related Problems}

\textbf{Simultaneous Matrix Completion.}
A mixed matrix is a matrix where each entry is either an element in $\fq$ or an indeterminate. 
Suppose there are $t$ mixed matrices, where each particular indeterminate can only appear once per matrix
but may appear in several matrices. The objective is to assign values for these indeterminates such 
that all resulting matrices simultaneously have maximum ranks. 
Such an assignment of values for the indeterminates is called a \emph{simultaneous completion}.
We can turn our problem into an instance of the Simultaneous Matrix Completion problem as follows. 
For a given $k\times n$ binary matrix $\bM$, let $\vM$ be the matrix obtained from $\bM$
by replacing $m_{i,j}$ with the indeterminate $\xi_{i,j}$ if $m_{i,j} = 1$. 
Consider the set of $t = \binom{n}{k}$ $k\times k$ submatrices of $\bM$. 
If we can find a simultaneous completion that make all $t$ resulting matrices full rank, 
then the $k\times n$ matrix $\bG$ obtained by this completion will be the generator matrix
of an MDS code and $\bG$ fits $\bM$.  
However, it was proved by Harvey {\et}~\cite{HarveyKargerYekhanin2006} that when $q \leq t = \binom{n}{k}$, 
a simultaneous completion may not exist. 
Furthermore, deciding whether a simultaneous completion exists is NP-complete.
Recall that in our conjecture, we require the field size to be as small as $n + k - 1$, which is much smaller than $\binom{n}{k}$. 
Therefore, reducing our problem to the Simultaneous Matrix Completion problem does not give us
any useful answer.  

\textbf{Weakly Secure Cooperative Data Exchange.}
In the Cooperative Data Exchange (CDE) problem, a group of wireless clients wants to exchange
a set of $n$ packets over a shared lossless broadcast channel. Each client has access to 
a subset of packets and requests all other packets. The objective is to find an optimal coding scheme,
which satisfies all requests in a minimum number of packet transmissions.
In their recent work~\cite{YanSprintson2013}, Yan and Sprintson studied the optimal coding scheme
that also achieves the maximum degree of secrecy in the following sense.  
The coding scheme is called $g$-secure if an eavesdropper who eavesdrops all of the
transmissions gains no information (in Shannon's sense) about any group of
$g+1$ or less original packets. 
The objective now is to design a $k \times n$ coding matrix that generates an $[n,k]_q$ MDS 
code, where $k$ is the optimal number of transmissions in the original CDE problem. 
As a result, they came up with a similar matrix completion problem
to ours and proposed a solution with large field size (exponential in the dimension of the matrix).
The existence of MDS code over small fields was also left as an open question. 
In our language, their problem description requires that the rows of the binary matrix $\bM$
(Fig.~\ref{fig:MDS_conjecture}) are partitioned into a certain number of groups or rows, 
and within each groups the rows have the same support. 
We show later, in the Appendix, that a coding matrix for CDE, in fact, must satisfy 
the MDS condition stated in our conjecture. Thus, their code design problem is equivalent to ours.  

\textbf{Distributed Reed-Solomon Codes.}
In a recent work of Halbawi {\et}~\cite{HalbawiHoYaoDuursma2013}, a simple multiple
access network (SMAN) was considered. A number of independent sources with arbitrary rates
try to convey all information to a single destination node, via $n$ relay nodes. 
The objective is to design an efficient coding scheme that can correct arbitrary/adversarial errors
on up to $z$ links/nodes. 
The authors~\cite{HalbawiHoYaoDuursma2013} proposed to use the so-called Distributed Reed-Solomon
codes, with the field size as small as $n + 1$, and also faced a similar matrix completion problem.
We prove in the Appendix that their code design problem is actually equivalent to ours. 
In our language, the authors~\cite{HalbawiHoYaoDuursma2013} proved that our GM-MDS
Conjecture holds for the case when the rows of $\bM$ are divided into \emph{three} groups
such that within each group, the rows share the same supports.  

\textbf{MDS Codes with Balanced Sparsest Generator Matrices.}
In our previous work~\cite{DauSongDongYuenISIT2013}, we prove the existence
of $[n,k]_q$ MDS codes that have balanced sparsest generator matrices for 
$q \geq \binom{n-1}{k-1}$. Such generator matrices have minimum row weights $n - k +1$ and moreover, all columns have approximately the same weights. 
The correctness of GM-MDS Conjecture would imply the existence of such 
MDS codes with balanced sparsest generator matrices over much smaller fields (as long as $q \geq n + k - 1$).  

\subsection{Our Contribution}
In our GM-MDS Conjecture, we unify and simplify the recently studied problems on MDS codes
with generator matrices having specified supports~\cite{YanSprintson2013,DauSongDongYuenISIT2013,
HalbawiHoYaoDuursma2013}. In contrast to previous works in~\cite{YanSprintson2013,
HalbawiHoYaoDuursma2013}, we explicitly and neatly describe the condition imposed on the
support of the generator matrix of such MDS codes, which we refer to as the \emph{MDS Condition}
(Fig.~\ref{fig:MDS_conjecture}). 
Furthermore, we no longer include the requirement that the rows of the matrix are divided into 
groups of the same supports, which may significantly simplify the study of the problem. 

Based on a novel generalization of the well-known Hall's marriage theorem, we propose a combinatorial approach
to attack the problem at hand and prove that our conjecture holds for a special case, 
where the sets of $0$-entries of the rows of $\bM$ share with each other (pairwise) at most one common elements. 
Numerical result shows that the conjecture holds for all $k \leq 7$ and for all $n \geq k$.
With this approach, we completely reduce the original problem to a combinatorial set problem.   

\subsection{Definitions and Notation}
We denote by $\fq$ the finite field with $q$ elements. 
Let $[n]$ denote the set $\{1,2,\ldots,n\}$. 
The \emph{support} of a vector $\bu = (u_1, \ldots, u_n) \in \fq^n$ is defined by
$\supp(\bu) = \{i \in [n]:\ u_i \neq 0\}$.
The (Hamming) \emph{weight} of $\bu$ is $|\supp(\bu)|$. 
We can also define weight and support of a row of a matrix over some finite
field, by regarding the matrix row as a vector over that field. 
Apart from Hamming weight, we also use other standard notions from coding theory such as minimum distance, 
generator matrices, linear $[n,k]_q$ codes, MDS codes, and GRS (for instance, see \cite{MW_S}). 
For a matrix $\bG = (g_{i,j}) \in \fq^{k \times n}$, the \emph{support matrix} of $\bG$ 
is a $k \times n$ binary matrix $\bM = (m_{i,j})$ where $m_{i,j} = 0$ if and only if $g_{i,j} = 0$.  

Our generalization of Hall's marriage theorem is presented in Section~\ref{sec:Hall}. 
We then describe our approach and findings in Section~\ref{sec:main}.
The paper is concluded in Section~\ref{sec:conclusion}.  

\section{A Generalization of Hall's Marriage Theorem}
\label{sec:Hall}

We first recall the famous Hall's marriage theorem. 

\begin{theorem}[Hall's theorem] 
For each $i \in [k]$ let $R_i$ be an arbitrary subset of $[k]$. Suppose that 
\[
|\cup_{i \in I} R_i | \geq |I|,
\]
for all nonempty subsets $I \subseteq [k]$. Then for every $i \in [k]$ there exists
a single representative $r_i \in R_i$ such that $r_i \neq r_{i'}$ whenever $i \neq i'$. 
\end{theorem} 

Theorem~\ref{thm:Hall_g} generalizes Hall's theorem to the case of multiple representatives. In this generalization, the sets of representatives are allowed to
overlap, but not too much. 
Note that when $n = k$, this theorem reduces to Hall's theorem.  

\begin{theorem}[Generalized Hall's theorem]
\label{thm:Hall_g} 
For each $i \in [k]$ let $R_i$ be an arbitrary subset of $[n]$ $(n \geq k)$. Suppose that 
\begin{equation} 
\label{eq:1}
|\cup_{i \in I} R_i | \geq n - k + |I|,
\end{equation} 
for all nonempty subsets $I \subseteq [k]$. Then for every $i \in [k]$ there exists
a subset $R'_i \subseteq R_i$ such that
\begin{itemize}
	\item $|\cup_{i \in I} R'_i | \geq n - k + |I|$, for all nonempty subsets $I \subseteq [k]$.
	\item $|R'_i| = n - k + 1$, for all $i \in [k]$. 
\end{itemize}
Moreover, such subsets $R'_i$ can be found in polynomial time. 
\end{theorem}
\begin{proof} 
To simplify the notation, for a set $I \subseteq [k]$ we use $R_I$ to denote the union $\cup_{i \in I}R_i$. 

Suppose that the sets $R_i$ satisfy (\ref{eq:1}). We keep removing the elements of these
sets while maintaining the MDS Condition (\ref{eq:1}). Assume that at some point, the removal
of any element in any set $R_i$ would make them violate (\ref{eq:1}). 
We prove that now the sets $R_i$ have cardinality precisely $n-k+1$, which concludes the
first part of the theorem. 

Suppose, for contradiction, that there exists $r \in [k]$ such that $|R_r| \geq n - k + 2$. 
Take $a$ and $b$ in $R_r$, $a \neq b$. For all $i \in [k]$, let 
\begin{equation} 
\label{eq:2}
R^a_i = 
\begin{cases}
R_i \setminus \{a\},& \text{ if } i = r,\\
R_i,& \text{ otherwise,} 
\end{cases}
\end{equation} 
and
\begin{equation} 
\label{eq:3}
R^b_i = 
\begin{cases}
R_i \setminus \{b\},& \text{ if } i = r,\\
R_i,& \text{ otherwise.} 
\end{cases}
\end{equation} 
According to our assumption, both of the two collections of sets $\{R^a_i\}_{i \in [k]}$
and $\{R^b_i\}_{i \in [k]}$ violate (\ref{eq:1}). 
Therefore, there exist two nonempty subsets $A \subseteq [k]$ and $B \subseteq [k]$, 
$r \notin A \cup B$, such that 
\begin{equation} 
\label{eq:4}
|R^a_{A \cup \{r\}}| < n - k + |A| + 1,
\end{equation} 
and
\begin{equation} 
\label{eq:5}
|R^b_{B \cup \{r\}}| < n - k + |B| + 1. 
\end{equation} 
Since $r \notin A$, by (\ref{eq:2}) we have
\begin{equation}
\label{eq:6}
|R^a_{A \cup \{r\}}| \geq |R^a_A| = |R_A| \geq n - k + |A|.  
\end{equation} 
Similarly, since $r \notin B$, by (\ref{eq:3}) we have
\begin{equation}
\label{eq:7}
|R^b_{B \cup \{r\}}| \geq |R^b_B| = |R_B| \geq n - k + |B|.  
\end{equation} 
From (\ref{eq:4}) and (\ref{eq:6}) we deduce that
\begin{equation}
\label{eq:8}
|R^a_{A \cup \{r\}}| = |R^a_A| = |R_A| = n - k + |A|.
\end{equation} 
Similarly, from (\ref{eq:5}) and (\ref{eq:7}) we have
\begin{equation}
\label{eq:9}
|R^b_{B \cup \{r\}}| = |R^b_B| = |R_B| = n - k + |B|.  
\end{equation} 
Therefore, 
\begin{equation} 
\label{eq:10}
R^a_{A \cup \{r\}} \cap R^b_{B \cup \{r\}} = R_A \cap R_B.
\end{equation} 
Moreover, as $a \in R^b_{B \cup \{r\}}$ and $b \in R^a_{A \cup \{r\}}$, we deduce that 
\begin{equation} 
\label{eq:11}
R^a_{A \cup \{r\}} \cup R^b_{B \cup \{r\}} = R_{A \cup B \cup \{r\}}. 
\end{equation} 
From (\ref{eq:8}) and (\ref{eq:9}) we have
\begin{equation} 
\label{eq:13}
\begin{split}
&\quad\ 2(n-k) + |A| + |B|\\
&= |R^a_{A \cup \{r\}}| + |R^b_{B \cup \{r\}}|\\
&= |R^a_{A \cup \{r\}} \cup R^b_{B \cup \{r\}}| + |R^a_{A \cup \{r\}} \cap R^b_{B \cup \{r\}}|\\ 
&= |R_{A \cup B \cup \{r\}}| + |R_A \cap R_B|,
\end{split} 
\end{equation}  
where the last transition is due to (\ref{eq:10}) and (\ref{eq:11}). 
We further evaluate the two terms of the last sum in (\ref{eq:13}) as follows. 
The first term
\begin{equation} 
\label{eq:14}
\begin{split} 
|R_{A \cup B \cup \{r\}}| &\geq n - k + |A \cup B \cup \{r\}|\\
&= n - k + |A \cup B| + 1. 
\end{split} 
\end{equation} 
The second term
\begin{equation} 
\label{eq:15}
|R_A \cap R_B| \geq n - k + |A \cap B|,
\end{equation} 
which can be explained below. 
\begin{itemize}
  \item If $A \cap B \neq \varnothing$, then by applying (\ref{eq:1}) to $A \cap B$ we obtain
	\[
	|R_A \cap R_B| \geq |R_{A \cap B}| \geq n - k + |A\cap B|. 
	\]
	\item If $A \cap B = \varnothing$, then $n - k + |A \cap B| = n - k$.
We have
\begin{equation} 
\label{eq:16}
R^a_{A \cup \{r\}} = R^a_A \cup R^a_r = R_A \cup (R_r \setminus \{a\}). 		
\end{equation} 
By (\ref{eq:8}), $R^a_{A \cup \{r\}} = R_A$. Combining this with (\ref{eq:16}) we deduce that
\begin{equation} 
\label{eq:17}
R_r \setminus \{a\} \subseteq R_A. 
\end{equation}  
Similarly, 
\begin{equation} 
\label{eq:18}
R_r \setminus \{b\} \subseteq R_B. 
\end{equation}  
From (\ref{eq:17}) and (\ref{eq:18}) we have
	\[
	|R_A \cap R_B| \geq |R_r \setminus \{a,b\}| \geq (n - k + 2) - 2 = n - k,   
	\]
which proves that (\ref{eq:15}) is correct when $A \cap B = \varnothing$. 
\end{itemize}
Finally, from (\ref{eq:13}), (\ref{eq:14}), and (\ref{eq:15}) we deduce that
\[
\begin{split} 
&\quad \ 2(n-k) + |A| + |B|\\
&\geq \big(n - k + |A \cup B| + 1\big) + \big(n - k + |A \cap B| \big)\\
&= 2(n-k) + |A| + |B| + 1,
\end{split} 
\]
which produces a contradiction. 

The proof of the first part of this theorem also provides a polynomial time
algorithm to find subsets of $R_i$'s that all have cardinality $n -k+1$ yet still
maintain the MDS Condition (\ref{eq:1}). Indeed, we keep removing the
elements of the subsets $R_i$ in the following way. If there exists $r \in [k]$
such that $|R_r| \geq n - k + 2$, then as we just prove, for $a, b \in R_r$, 
it is impossible that removing $a$ or $b$ from $R_r$ both render the MDS Condition
violated. Therefore, we can either remove $a$ or $b$ while still maintaining
the MDS Condition. Note that according to \cite{DauSongDongYuenISIT2013}, 
the MDS Condition can be verified in polynomial time in $k$ and $n$. 
Therefore, this algorithm terminates in polynomial time in $k$ and $n$ 
and produces subsets $R'_i$'s of the original sets $R_i$'s that satisfy the
stated requirement in the theorem.  
\end{proof}  
\vskip 10pt 

\section{A Combinatorial Approach to the GM-MDS Conjecture}
\label{sec:main}

Our main idea is to first simplify the GM-MDS Conjecture, using Theorem~\ref{thm:Hall_g}. 
Then by employing generalized Reed-Solomon (GRS) codes, we reduce our code design
problem over low field sizes to a pure combinatorial set problem. 
Our main findings include 
\begin{itemize}
	\item a theoretical proof of the correctness of the Simplified GM-MDS Conjecture 
	when the sets of $0$-entries of rows of $\bM$ intersect each other (pairwise) at 
at most one element, 
\item a numerical proof of the correctness of the Simplified GM-MDS Conjecture
for all $k \leq 7$. 
\end{itemize}

\subsection{Simplified GM-MDS Conjecture}
Based on Theorem~\ref{thm:Hall_g}, we can simplify the GM-MDS Conjecture
to the case where the row weights of the given matrix $\bM$ are precisely $n-k+1$. 
\begin{figure}[htb]
\centering{
\fbox{
\parbox{3.3in}{
\centerline{\bf Simplified GM-MDS Conjecture}
The statement is the same as in the GM-MDS Conjecture, except that we assume
all rows of $\bM$ have weight precisely $n - k + 1$. 
}
}
}
\caption{The Simplified GM-MDS Conjecture}
\label{fig:simplified_MDS_conjecture}
\end{figure}
The two conjectures are, in fact, equivalent. Indeed, if the GM-MDS Conjecture
is true, then obviously the Simplified GM-MDS Conjecture is also true. 
Conversely, suppose that the Simplified GM-MDS Conjecture is true, we need
to show that the GM-MDS Conjecture is also true. Let $\bM$ be any 
$k \times n$ binary matrix that satisfies the MDS Condition. By applying 
Theorem~\ref{thm:Hall_g} to the supports $R_i$'s of rows of $\bM$, we can
find another $k\times n$ binary matrix $\bM'$ that fits $\bM$, satisfies the MDS Condition, 
and furthermore have row weights precisely $n-k+1$. As we assume that the Simplified GM-MDS
Conjecture is true, as long as $q \geq n + k - 1$ there exists an $[n,k]_q$ MDS code 
with a generator matrix that fits $\bM'$, and hence, also fits $\bM$. 
Thus, the two conjectures are equivalent.  

\subsection{Reduction of the Simplified GM-MDS Conjecture to a Set Problem}

Let $\bM$ be a $k\times n$ binary matrix that satisfies the MDS Condition and 
has row weights $n - k + 1$. We aim to show that there exists an $[n,k]_q$
generalized Reed-Solomon (GRS) code that has a generator matrix $\bG$ fitting $\bM$.
As all rows of $\bM$ have weight $n - k + 1$, in fact, $\bM$ is the support
matrix of $\bG$, i.e. $g_{i,j} = 0$ if and only if $m_{i,j} = 0$. 
Let $\al_1, \al_2, \ldots, \al_n$ be $n$ distinct elements of $\fq$, the evaluation points
in the standard construction of a GRS code. 
Since a codeword of weight $n - k + 1$ in an $[n,k]_q$ MDS code
is uniquely determined (up to scalar multiple) by its support (\cite[Ch. 11]{MW_S}), 
the rows of a generator matrix $\bG$ (with support $\bM$) of the desired GRS code
correspond to the polynomials 
\begin{equation} 
\label{eq:19}
f_i(x) = \prod_{j \in Z_i} (x-\al_j),\quad i \in [k],
\end{equation}
where 
\begin{equation} 
\label{eq:12}
Z_i = [n] \setminus \supp(\bM_i) = \{j\in [n] \mid m_{i,j} = 0\}.
\end{equation}   
In other words, 
\begin{equation} 
\label{eq:G}
\bG = 
\begin{pmatrix}
f_1(\al_1) & f_1(\al_2) & \cdots & f_1(\al_n)\\
f_2(\al_1) & f_2(\al_2) & \cdots & f_2(\al_n)\\
\vdots & \vdots & \cdots & \vdots \\
f_k(\al_1) & f_k(\al_2) & \cdots & f_k(\al_n)\\ 
\end{pmatrix}.
\end{equation} 
For $i \in [k]$ let
\begin{equation} 
\label{eq:20}
f_i(x) = \sum_{j = 1}^k a_{i,j} x^{j-1}.
\end{equation} 
Then we can rewrite $\bG$ as
\[
\begin{split} 
\bG &= \bA \bV\\
&=
\begin{pmatrix}
a_{1,1} & a_{1,2} & \cdots & a_{1,k}\\
a_{2,1} & a_{2,2} & \cdots & a_{2,k}\\
\vdots & \vdots & \cdots & \vdots \\
a_{k,1} & a_{k,2} & \cdots & a_{k,k}
\end{pmatrix}
\begin{pmatrix}
1 & 1 & \cdots & 1\\
\al_1 & \al_2 & \cdots & \al_n\\
\al_1^2 & \al_2^2 & \cdots & \al_n^2\\
\vdots & \vdots & \cdots & \vdots \\
\al_1^{k-1} & \al_2^{k-1} & \cdots & \al_n^{k-1}
\end{pmatrix}.
\end{split} 
\]
Clearly, $\bG$ has full rank and hence is truly a generator matrix of an $[n,k]_q$ GRS
code if and only if $\bA$ is invertible. Let
\[
F(\al_1,\ldots,\al_n) = \det(\bA)\prod_{i > j}(\al_i - \al_j), 
\]
which is a polynomial in $\al_1, \ldots, \al_n$. 

\vskip 10pt 
\begin{lemma}
\label{lem:1}
The highest degree of each variable $\al_t$ in $F(\al_1,\ldots,\al_n)$ is at most $n + k - 2$. 
\end{lemma} 

\begin{proof}
We have
\begin{equation} 
\label{eq:22}
\det(\bA) = \sum_{\sigma \in \mathscr{S}_k} \sgn(\sigma) \prod_{i=1}^k a_{i,\sigma(i)},
\end{equation}  
where $\mathscr{S}_k$ denotes the symmetric (permutation) group on $k$ elements, and
$\sgn(\sigma)$ denotes the sign of the permutation $\sigma$. 
According to (\ref{eq:19}) and (\ref{eq:20}),
\begin{equation} 
\label{eq:23}
a_{i,j} = 
\begin{cases}
(-1)^{k-j} \sum_{T \subseteq Z_i,\ |T| = k-j} \prod_{t \in T}\al_t,& \text{ if } 1\leq j < k,\\
1,& \text{ if } j = k, 
\end{cases}  
\end{equation} 
where $Z_i$ is given in (\ref{eq:12}). 
Therefore, the highest degree of $\al_t$ in each $a_{i,j}$ is at most one, for 
every $t \in [n]$, $i \in [k]$ and $j \in [k]$. As a result, by (\ref{eq:22})
the highest degree of $\al_t$ in $\det(\bA)$ is at most $k - 1$. 
Obviously the highest degree of $\al_t$ in $\prod_{i > j}(\al_i - \al_j)$ is $n - 1$. 
Thus, the highest degree of $\al_t$ in $F(\al_1, \ldots, \al_n)$ is at most $n + k - 2$.    
\end{proof} 

\begin{lemma}
\label{lem:2}
If the polynomial $F(\al_1,\ldots,\al_n)$ is not identically zero then as long as $q \geq n + k - 1$, 
there exist $\al^*_1, \ldots, \al^*_n$ in $\fq$ such that $\bG$ 
given by (\ref{eq:G}), with $\al_i$ being replaced by $\al^*_i$, 
is a generator matrix of an $[n,k]_q$ GRS code. 
\end{lemma}
\begin{proof} 
By Lemma~\ref{lem:1}, the highest degree of each $\al_i$ in $F(\al_1,\ldots,\al_n)$
is at most $n + k - 2$. 
According to~\cite[Lemma 4]{Ho2006}, if $F(\al_1,\ldots,\al_n)$ is not identically 
zero then provided that $q > n + k - 2$, there exist $\al^*_1, \ldots, \al^*_n$ in $\fq$ such that
$F(\al^*_1,\ldots,\al^*_n) \neq 0$. Hence $\al^*_i \neq \al^*_j$ whenever $i \neq j$. 
Moreover, as 
\[
\det(\bA)|_{\al^*_1,\ldots,\al^*_n} \neq 0,
\]
the matrix $\bG = \bA \bV$ also has full rank and hence is a generator matrix of an $[n,k]_q$
GRS code.  
\end{proof} 

\begin{figure}[htb]
\centering{
\fbox{
\parbox{3.3in}{
\centerline{\bf Unique-Multiset Conjecture}
Let $Z_1, Z_2, \ldots, Z_k$ be $(k-1)$-subsets of $[n]$ that satisfy
\[
|\cap_{i \in I} Z_i | \leq k - |I|,
\]  
for every nonempty subset $I \subseteq [k]$. Consider all permutations $\sigma \in \mathscr{S}_k$ of $[k]$
and consider all possible ways to select some $(\sigma(i) - 1)$-subset $S_i$ of the set $Z_i$
for each $i \in [k]$. Take multiset union of these $k$ subsets $S_i$.
Then there exists one of such unions that is unique among all choices of permutations
$\sigma$ and all choices of subsets $S_i$. 
}
}
}
\caption{The Unique-Multiset Conjecture}
\label{fig:unique-multiset_conjecture}
\end{figure}

According to Lemma~\ref{lem:2}, the key point is to show that if $\bM$ satisfies 
the MDS Condition then $\det(\bA)$ is not identically zero. At this moment we 
are not able to prove this statement in general. However, we pose yet another conjecture
on a pure combinatorial problem, referred to as the Unique-Multiset Conjecture. 
We prove in Lemma~\ref{lem:3} that if this conjecture is true, then so is the Simplified GM-MDS Conjecture. The relation among the conjectures raised
in this paper is illustrated in Fig.~\ref{fig:diagram}.   
We show in Theorem~\ref{thm:case1} that the Unique-Multiset Conjecture is in fact true for a nontrivial instance. 
Finally, the numerical result confirms the correctness of this conjecture up to $k = 7$. 

\begin{figure}[htbp]
	\centering
		\includegraphics{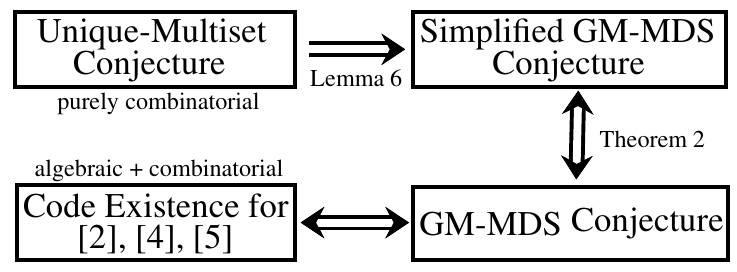}
	\caption{Relation among the conjectures}
	\label{fig:diagram}
\end{figure}

\begin{example}
\label{ex:1}
We consider an example to illustrate the Unique-Multiset Conjecture. Let $k = 3$, $n = 6$. Both Choice 1 and Choice 2 correspond to the identity permutaion. However, the selection of $S_i$ in each choice is different from the other. Choice 3 corresponds to $\sigma(1) = 1, \sigma(2) = 3, \sigma(3) = 2$. 
\begin{figure}[htb]
	\centering
		\begin{tabular}{|l|l|l|l|}
		\hline
		& Choice 1 & Choice 2 & Choice 3\\
		\hline 
			$Z_1 = \{5,6\}$ &$S_1=\{\}$ & $S_1=\{\}$ &  $S_1=\{\}$ \\
			$Z_2 = \{1,4\}$ & $S_2=\{4\}$ & $S_2=\{1\}$ &  $S_2=\{1,4\}$\\
			$Z_3 = \{3,4\}$ & $S_3=\{3,4\}$ &$S_3=\{3,4\}$  & $S_3=\{3\}$\\
			\hline
			multiset union & $\{3,4,4\}$ & $\{1,3,4\}$  & $\{1,3,4\}$\\
		\hline
		\end{tabular}
\end{figure}
In total there are $12 = 3!\times 2$ different choices, but we only list three of them here. 
Choice 1 produces $\{3,4,4\}$, which is unique (easy to verify). By contrast, $\{1,3,4\}$
is not a unique multiset union because it can be obtained by different choices of permutation
and subsets. 
\end{example}

\begin{lemma}
\label{lem:3}
If the Unique-Multiset Conjecture holds then so does the Simplified GM-MDS Conjecture. 
\end{lemma}
\begin{proof}
Let $\bM$ be a $k \times n$ binary matrix that has row weights $n - k + 1$ and satisfies
the MDS Condition. Let $Z_i = [n] \setminus \supp(\bM_i)$, then $|Z_i| = k - 1$ $(i \in [k])$. 
Using the De Morgan's law, $\bM$ satisfies the
MDS Condition if and only if 
\[
|\cap_{i \in I} Z_i | \leq k - |I|,
\]  
for every nonempty subset $I \subseteq [k]$. 
According to (\ref{eq:22}) and (\ref{eq:23}), if we treat $\det(\bA)$ as a polynomial, 
which is a sum of monomials in $\al_1, \ldots, \al_n$, then there is a one-to-one
correspondence between the monomials and the multisets described in the Unique-Multiset
Conjecture. More specifically, the monomial $\al_1^{p_1}\al_2^{p_2} \cdots \al_n^{p_n}$
corresponds to the multiset
\[
\big\{\underbrace{11 \cdots 1}_{p_1}\underbrace{22 \cdots 2}_{p_2}\cdots \underbrace{nn \cdots n}_{p_n} \big\}.
\]
Therefore, if the Unique-Multiset Conjecture is correct than there is a monomial in the 
expression of $\det(\bA)$ that appears exactly once. Hence, $\det(\bA)$ is not identically zero.
By Lemma~\ref{lem:2}, as long as $q \geq n + k - 1$, there exists an $[n,k]_q$ GRS code
that has a generator matrix fitting $\bM$. Hence the Simplified GM-MDS Conjecture holds. 
\end{proof}

\subsection{Results on the Unique-Multiset Conjecture}

\begin{theorem}
\label{thm:case1}
The Unique-Multiset Conjecture holds when the sets $Z_i$'s satisfy an additional property 
that $|Z_i \cap Z_{i'}| \leq 1$ for all $i \neq i'$.  
\end{theorem} 
\begin{proof} 
With the additional property that $|Z_i \cap Z_{i'}| \leq 1$ for all $i \neq i'$, the MDS
Condition now only requires
\begin{equation} 
\label{eq:24}
\cap_{i \in [k]}Z_i = \varnothing. 
\end{equation} 
We first construct a multiset according to the rule in the conjecture, and then prove that
it is unique. Note that if $Z_2 \cap Z_1 = \{j\}$, then there must exist some $i \in [k]\setminus \{1,2\}$
so that $j \notin Z_2 \cap Z_i$, for otherwise (\ref{eq:24}) would be violated. 
Therefore, reordering $Z_i$'s if necessary, we can suppose that $Z_2 \cap Z_3 \not\subseteq Z_1$
if $Z_2 \cap Z_3 \neq \varnothing$.  
Let the permutation $\sigma$ be the identity permutation, i.e. $\sigma(i) = i$ for all $i \in [k]$. 
We construct the $(i-1)$-subsets $S^*_i$ of $Z_i$ $(i \in [k])$ as follows.
\begin{itemize}
	\item \textbf{Step} $1$: $S^*_1 = \varnothing$;
	\item \textbf{Step} $2$: $S^*_2 = \{j'\}$, where $j'$ is the element in $(Z_2 \cap Z_3) \setminus Z_1$, if $Z_2 \cap Z_3 \neq \varnothing$, 
	and $j'$ is an arbitrary element in $Z_2 \setminus Z_1$ if $Z_2 \cap Z_3 = \varnothing$; 
  \item \textbf{Step} $i$: $(3 \leq i < k)$ The $i - 1$ elements of $S^*_i$ are selected as follows. 
	First, for each $i' < i$, $i' \geq 2$, we include in $S^*_i$ the common element (if any) of $Z_i$ and $Z_{i'}$. Note that
	there is at most one such common element for each $2 \leq i' < i$. Hence there are at most $i-2$ such elements.
	Second, we include in $S^*_i$ the common element (if any) of $Z_i$ and $Z_{i+1}$, if that element was not
	included earlier. There is at most one such element.  
	Finally, we fill in $S^*_i$ with other elements of $Z_i$ arbitrarily so that $|S^*_i| = i - 1$. 
	\item \textbf{Step} $k$: $S^*_k = Z_k$. 
\end{itemize}
For example, for $k = 4$, $Z_1 = \{1,2,3\}$, $Z_2 = \{1,4,5\}$, $Z_3 = \{2,4,6\}$, $Z_4 = \{3,5,7\}$, 
the sets $S^*_i$ are selected as follows. In Step 1, $S^*_1 = \varnothing$. In Step 2, as $4 \in Z_2 \cap Z_3 \setminus Z_1$, 
we have $S^*_2 = \{4\}$. In Step 3, we first include in $S^*_3$ the common element $4 \in Z_3 \cap Z_2$. 
Note that $Z_3 \cap Z_4 = \varnothing$. Hence we take an arbitrary element in $Z_3\setminus \{4\}$, for instance $6$, 
and have $S^*_3 = \{4,6\}$. In the last step, $S^*_4 = Z_4 = \{3,5,7\}$. 
\begin{table}[htb]
	\centering
		\begin{tabular}{|l|l|}
		\hline
			$Z_1 = \{1,2,3\}$ & $S^*_1 = \{\}$\\
			\hline
			$Z_2 = \{1,4,5\}$ & $S^*_2 = \{4\}$\\
			\hline
			$Z_3 = \{2,4,6\}$ & $S^*_3 = \{4,6\}$\\
			\hline
			$Z_4 = \{3,5,7\}$ & $S^*_4 = \{3,5,7\}$\\
			\hline
		\end{tabular}
		\caption{An example on how to construct the unique multiset.}
		\label{tab:1}
\end{table}
The resulting multiset is $[4,4,6,3,5,7]$. 

We now prove that the multiset union $\M^*$ of the sets $S^*_1, \ldots, S^*_k$ constructed above is unique 
among all choices of $\sigma \in \mathscr{S}_k$ and all choices of $(\sigma(i)-1)$-subset $S_i$ of $Z_i$ $(i \in [k])$. 
Arrange all the elements of $\M^*$ in a $k$-row array $\A^*$, so that the $i$th row of $\A$
consists of the elements of $S^*_i$ $(i \in [k])$. The multiset $\M^*$ is unique as required
if and only if any $k$-row array $\A$ that satisfies
\begin{itemize}
	\item \textbf{(C1)} precisely one row of $\A$ has $i-1$ entries, for each $i \in [k]$,
	\item \textbf{(C2)} all entries in the $i$th row of $\A$ are elements of $Z_i$, for each $i \in [k]$, 
	\item \textbf{(C3)} the multiset union of the rows of $\A$ is $\M^*$, 
\end{itemize}
would be identical to $\A^*$. 
If such an array exists, it would be obtained from $\A^*$ via a sequence (possibly empty) of \emph{entry-moves}. 
An entry-move, or just a move, for short, is the act of moving an entry from one row to another. 
We only consider \emph{valid} moves in which the entry $j$ in row $i$ is moved to 
row $i'$ under the following rules
\begin{itemize}
	\item \textbf{(R1)} $j \in Z_{i'}$, 
	\item \textbf{(R2)} $j$ is not an entry in the row $i'$ before the move, i.e. the move does not
	create duplicated entries in row $i'$. 
\end{itemize} 
We also assume that there is \emph{no redundant entry-move} in the sequence in the following 
sense. If there is an entry-move that moves $j$ out of row $i$ then there is no other
entry-move in the sequence that moves $j$ from anywhere to row $i$.
Such a sequence is called \emph{irredundant}. 
We now examine all possible valid entry-moves on $\A^*$ and show that if $s$ is
an \emph{irredundant} sequence of valid moves that turns $\A^*$ into another array $\A$
satisfying the aforementioned conditions (C1)-(C3), then $\A \equiv \A^*$ and $s$ is empty. 
Note that as the sequence is irredundant, if $j \in S^*_i$ then in the sequence, there is
no entry-move that moves $j$ from another row to row $i$.  
\\

\nin\textbf{Claim 1:} For each $i \leq k$, $i \geq 2$, there is no valid move from any row $i' < i$
to row $i$. In other words, there is no valid \emph{downward} move in the array $\A^*$. 
\begin{proof}[Proof of Claim 1]
First, as the first row of $\A^*$ has no entries, there is obviously no valid move from the first row to
any other row. 
Note that as the sequence of entry-moves is irredundant, all the entries that are moved to
the first row cannot be moved out of this row. 
According to the way $S^*_i$ is constructed in Step $i$, the set $S^*_i$
contains all common elements (if any) of $Z_i$ and $Z_{i'}$ for $2 \leq i' < i$.
Therefore, according to (R1) and (R2), whether $Z_i \cap Z_{i'} = \varnothing$
or not, there is no valid move from row $i'$ to row $i$.   
\end{proof}
\vskip 10pt

\nin\textbf{Claim 2:} For each $i \leq k - 1$ and $i+1 \leq \ell \leq k$, 
there are at most $\ell - i - 1$ valid moves (\emph{upward}) from the rows $i+1, i+2, \ldots, \ell$ to row $i$. 
\begin{proof}[Proof of Claim 2]
As $|Z_i \cap Z_{i'}| \leq 1$ when $i' \neq i$, and as the sequence of entry-moves 
is irredundant, there is at most one valid move from each of the
rows $i+1, i+2, \ldots, \ell$ to row $i$. 
For $i = 1$, according to Step 2, $S^*_2$ contains no element in $Z_1$. Hence, 
by (R1), there is no valid move from row $2$ to row $1$.  
For $i \geq 2$, according to Step 2 and Step $i$, $S^*_i$ and $S^*_{i+1}$
both contain their unique common element in $Z_i \cap Z_{i+1}$ (if any). Therefore, according to (R1) and (R2), 
in any case whether $Z_i \cap Z_{i+1} = \varnothing$ or not, 
there is no valid move from row $i+1$ to row $i$. Thus, there are at most $\ell-i-1$
valid moves from the rows $i + 1, i + 2, \ldots,\ell$ to row $i$.   
\end{proof} 
\vskip 10pt

We now prove in an induction manner that if an irredundant sequence $s$ of valid moves
turn $\A^*$ in $\A$ that satisfies (C1)-(C3), then it must be an empty sequence, and hence 
$\A \equiv \A^*$. 

According to Claim 1 and Claim 2, in total there are at most $k - i - 1$ valid moves
from other rows to row $i$, for each $i \leq k - 1$. As $|S^*_i| = i - 1$, after applying
the sequence $s$ to $\A^*$, row $i$ has at most 
\[
(i - 1) + (k - i - 1) = k - 2 < k - 1 
\]   
entries. 
Therefore, by (C1), the $k$th rows of $\A$ and $\A^*$ are identical. 
Hence, there is no move in $s$ originating from the row $k$. 

Again, according to Claim 1 and Claim 2, in total there are at most $(k-1) - i - 1$ valid moves
from other rows (row $k$ is excluded as $s$ contains no moves from this row) to row $i$, 
for each $i \leq k - 2$. As $ |S^*_i| = i - 1$, after applying the sequence $s$ to $\A^*$, row $i$ has at most 
\[
(i - 1) + (k - i - 2) = k - 3 < k - 2
\]   
entries. Therefore, by (C1), the $(k-1)$th rows of $\A$ and $\A^*$ are identical. 
Hence, there is no move in $s$ originating from either row $k$ or $k-1$.
Repeating this argument again and again until we establish that the sequence $s$
leaves all rows of $\A^*$ unchanged, we complete the proof of the theorem.  
\end{proof} 
\vskip 10pt

According to Lemma~\ref{lem:3} and Theorem~\ref{thm:case1}, we settle the Simplified 
GM-MDS Conjecture for the case when the matrix $\bM$ satisfies an additional property
that the set of zero coordinates of rows of $\bM$ intersect each other at at most one element. 

On the computational side, we verified that the Unique-Multiset Conjecture holds, and hence
so does the (Simplified) GM-MDS Conjecture, for all $k \leq 7$. We ran a program to test
all legitimate input matrices $\bM$ for all $n \leq k(k-1)$. Note that by examining the 
statement of the Unique-Multiset Conjecture, clearly it is sufficient to verify the conjecture 
for all $n \leq k(k-1)$ only.    

\section{Conclusion}
\label{sec:conclusion}

We unify the recently studied problems on designing MDS codes given certain
constraints on the support of the generator matrices and propose 
a combinatorial approach that reduces the whole problem to an elementary set problem. 
We report some initial progress on this promising approach that provides
further evidences on the solvability of the problem at hand, which we believe
is of interest to the coding theory community. 
If the conjecture in this work is proved to be true, then the existence of 
valid codes for the applications described in~\cite{YanSprintson2013, DauSongDongYuenISIT2013,
HalbawiHoYaoDuursma2013} will be confirmed. 
In such a case, a randomized algorithm to find such good codes, according to
the proposed approach, is already available. 
Designing an efficient algorithm that deterministically find the good codes would be the next
task to consider. 

\section*{Acknowledgment}
The authors thank Yeow Meng Chee and Muriel M\'{e}dard for helpful discussions. 

\bibliographystyle{IEEEtran}
\bibliography{Existence-GRS-Given-GM-Support}

\section*{Appendix}

\subsection{Relation With the Weakly Secure Cooperative Data Exchange Problem}
We show below that the support matrix of the coding matrix used for a
Cooperative Data Exchange (CDE) instance in Yan and 
Sprintson's work~\cite{YanSprintson2013} must satisfy the MDS Condition. 
As a consequent, the correctness of the GM-MDS Conjecture would imply the
existence of an optimal CDE coding scheme over small field size $(q \geq n + k - 1)$
that has maximum degree of secrecy~\cite{YanSprintson2013}. 
Conversely, it is straightforward that our problem stated in the GM-MDS Conjecture
corresponds to a special case of their problem when in an optimal CDE coding scheme, 
each client broadcasts at most one coded packet. 

Let $X = \{x_1,x_2,\ldots, x_n\}$ be the $n$ original packets. Let $H_s \subsetneq [n]$ $(s \in [m])$ be 
the set of indices of original packets in $X$ that Client $C_s$ has access to. 
Let $b_s$ $(s \in [m])$ be the number of coded packets broadcast from $C_s$ according to
an optimal CDE coding scheme. 
Let $k = \sum_{s \in [m]} b_s$, which is the optimal total number of transmissions. 
Then the key problem in~\cite{YanSprintson2013} is to find a $k \times n$ matrix $\bG = (g_{i,j})$
that is a generator matrix of an $[n, k]$ MDS code over small fields. Moreover, it is required
that $\bG$ must fit a binary matrix $\bM = (m_{i,j})$, i.e. $g_{i,j} = 0$ whenever $m_{i,j} = 0$, which is 
defined as follows. The $k$ rows of $\bM$ are divided into $m$ groups, where the first group
consists of the first $b_1$ rows, the second group consists of the next $b_2$ rows, and so forth. 
Moreover, for each $s \in [m]$, all $b_s$ rows of $\bM$ in the corresponding group
have the same support, which is exactly $H_s$. We prove below that such a matrix $\bM$,
in fact, satisfies the MDS Condition stated in the GM-MDS Conjecture (Fig.~\ref{fig:MDS_conjecture}). 
 
According to Courtade and Wesel~\cite{CourtadeWesel-Milcom-2010}, we have
\[
|\cap_{s \in S} \overline{H}_s| \leq \sum_{s \notin S} b_s,
\] 
for every set $S \subseteq [m]$. The inequality is obtained by a typical cut-set argument, which
 requires that for every set of clients,
the coding scheme must provide them with at least as many packets from other clients 
as the number of original 
packets that the clients in the set all want. We can rewrite the inequality above as
\[
|\cup_{s \in S} H_s| \geq  n - \sum_{s \notin S} b_s, 
\]  
which in turn is equivalent to 
\begin{equation} 
\label{eq:25}
|\cup_{s \in S} H_s| \geq  n - k + \sum_{s \in S} b_s. 
\end{equation}  
Let $I$ be an arbitrary nonempty subset of $[k]$. Let $S$ be a subset of $[m]$
where $s \in S$ if and only if $I$ contains at least one row from the $s$th group of rows of $\bM$
(described earlier). Then 
\begin{equation} 
\label{eq:26}
\sum_{s \in S} b_s \geq |I|.
\end{equation} 
Note that $H_s$ is precisely the supports $R_i$ of the rows $i$ in the $s$th group of $b_s$ rows of $\bM$.
Therefore, according to (\ref{eq:25}) and (\ref{eq:26}), 
\[
\begin{split} 
|\cup_{i \in I} R_i| &= |\cup_{s \in S} H_s|\\
&\geq n - k + \sum_{s \in S} b_s\\
&\geq n - k + |I|. 
\end{split} 
\]
Thus, $\bM$ satisfies the MDS Condition. 

\subsection{Relation With the Distributed Reed-Solomon Code Problem}

In a Simple Multiple Access Network (SMAN)~\cite{HalbawiHoYaoDuursma2013}, 
there are $m$ sources $S_1, \ldots, S_m$ with information rates $r_1, \ldots, r_m$. 
These sources transmit all information they have to a single destination via a set of $n$ relays
$V_1, \ldots, V_n$. 
Each source $S_i$ is connected to a subset of relays via $r_i$-capacity links. 
Each relay $V_j$ is connected to the destination via an unit-capacity link. 
There are at most $z$ link/node errors across the network. 
The cut-set bound requires that the total rate of each group of sources
is bounded from above by the min-cut between that group of sources and the destination in the network
less $2z$. 
The key problem is to design an appropriate coding matrix over small fields 
that guarantees the successful retrieval of all information by destination under $z$ errors,
whenever the rate vector satisfies the cut-set bound. 
The code design problem posed by Halbawi {\et}~\cite{HalbawiHoYaoDuursma2013}, 
as we show below, is equivalent to our problem, stated in the GM-MDS Conjecture 
(Fig.~\ref{fig:MDS_conjecture}). 

Let $\bA = (a_{i,j})$ be the $m \times n$ adjacency matrix for the SMAN instance as defined in~\cite{HalbawiHoYaoDuursma2013}. Then $a_{i,j} = 1$ if the source $S_i$ is connected to
the relay $V_j$ and zero otherwise.
The cut-set bound stated in~\cite{HalbawiHoYaoDuursma2013} can be restated as follows
\begin{equation} 
\label{eq:27}
|\cup_{i \in I} \supp(\bA_i) | \geq \sum_{i \in I} r_i + 2z, 
\end{equation}   
for every nonempty subset $I \subseteq [m]$, where $\supp(\bA_i)$ denotes the support of 
the $i$th row of $\bA$.  

Let $r_{\I} \define \sum_{i \in m}r_i$.
Let $\bM' = (m'_{i,j})$ be the $r_{\I}\times n$ binary matrix constructed as follows. 
\begin{itemize}
\item The first $r_1$ rows of $\bM'$ are the same as the first row of $\bA$.
\item The next $r_2$ rows of $\bM'$ are the same as the second row of $\bA$.
\item $\cdots$
\item The last $r_m$ rows of $\bM'$ are the same as the last row of $\bA$. 
\end{itemize}
Then by (\ref{eq:27}), it is easy to verify that
\begin{equation} 
\label{eq:28}
|\cup_{i \in I} \supp(\bM'_i) | \geq |I| + 2z, 
\end{equation} 
for every nonempty subset $I \subseteq [r_{\I}]$. 
Let $k \define n - 2z$. 
The main problem in~\cite{HalbawiHoYaoDuursma2013} is to design
a matrix $\bG'$ that fits $\bM'$ and at the same time generates an $r_{\I}$-dimensional
subspace of an $[n, k]_q$ GRS code, with small $q$ ($q \geq n + 1$ in their 
original paper).   

Setting $I = [m]$ in (\ref{eq:27}), we have
\[
n \geq |\cup_{i \in [m]} \supp(\bA_i)| \geq r_{\I} + 2z, 
\]
which in turn implies that 
\[
k = n - 2z \geq r_{\I}. 
\]
Let $\bM = (m_{i,j})$ be the $k \times n$ binary matrix 
obtained from $\bM'$ by appending to this matrix $k - r_{\I}$ all-one rows.
Then according to (\ref{eq:28}) and the construction of $\bM$, 
\[
|\cup_{i \in I} \supp(\bM_i) | \geq |I| + 2z = n - k + |I|, 
\] 
for every nonempty subset $I \subseteq [k]$. 
Hence $\bM$ satisfies the MDS Condition stated in the GM-MDS Conjecture. 
If the GM-MDS Conjecture holds, then there exists a $k \times n$ matrix
$\bG$ that fits $\bM$ and generates an $[n, k]_q$ GRS code provided that $q \geq n + k - 1$. 
Let $\bG'$ be the submatrix of $\bG$ that consists of the first $r_{\I}$
rows. Then clearly $\bG'$ fits $\bM'$ and generates an $r_{\I}$-dimensional
subspace of an $[n, k]_q$ GRS code, as desired. Thus, if our conjecture holds
then there exists a linear code for SMAN~\cite{HalbawiHoYaoDuursma2013} 
for every rate vector in the capacity region, with the field size as small as $n + k - 1$. 

Conversely, as our problem corresponds to the case where all rates are equal
$r_1 = r_2 = \cdots = r_m = 1$ and $r_{\I} = n - k$, our problem is a special case
of the problem in~\cite{HalbawiHoYaoDuursma2013}. In conclusion, the two
code design problems are essentially equivalent.


\end{document}